\documentclass{eptcs}

\usepackage{microtype}
\usepackage{multirow}
\usepackage{subcaption}
\usepackage{tikz}
\usetikzlibrary{shapes, positioning, patterns,decorations.pathreplacing}\tikzstyle{vertex}=[circle, draw, inner sep=0pt, minimum size=4pt,outer sep = 1pt]

\usepackage{booktabs} 
\usepackage{float}
\usepackage{boxedminipage}
\usepackage{times}
\usepackage{color}
\usepackage[linesnumbered,boxed]{algorithm2e}
\usepackage{graphicx}
\usepackage{latexsym}
\usepackage{hyperref}
\usepackage{amsmath,amsthm}
\usepackage{amssymb}
\usepackage[disable=True]{todonotes}
\usepackage[draft,author=]{fixme}
\usepackage[T1]{fontenc}
\usepackage{needspace}

	\usepackage{todonotes}
	\usepackage{enumitem}
	\usepackage{graphicx}
	\newcommand{\mypar}[1]{\smallskip
		
	\noindent\textbf{#1}}

\newtheorem{proposition}{Proposition}
\newtheorem{observation}{Observation}

\newenvironment{appdxProp}[1]{
{\it Proposition \ref{#1}}\\%
}
{
\smallskip
}

\newcommand{\ruleatom}{\mathop{{:}-}}
\newcommand{\wc}{\mathop{:\sim}}

\usepackage{wrapfig}

\newenvironment{myitemize}
{\begin{list}{$-$}{%
\setlength{\topsep}{2pt} 
\setlength{\leftmargin}{0pt} 
\setlength{\itemindent}{10pt}}} 
{\end{list}}

\newif\ifappendix\appendixtrue
\appendixfalse

\title{Deontic Paradoxes in ASP with Weak Constraints\footnote{ Work partially supported by the WWTF project TAIGER (ICT22-026).}}
\author{Christian Hatschka\qquad\qquad Agata Ciabattoni\qquad\qquad Thomas Eiter
\institute{Institute of Logic and Computation, TU Wien, Vienna, Austria}
\email{firstname.lastname@tuwien.ac.at}
}

\newcommand{\titlerunning}{Deontic Paradoxes in ASP}
\newcommand{\authorrunning}{C.Hatschka, A.Ciabattoni \& T.Eiter}

\hypersetup{
  bookmarksnumbered,
  pdftitle    = {\titlerunning},
  pdfauthor   = {\authorrunning},
  pdfsubject  = {Deontic Paradoxes in ASP},               
  pdfkeywords = {Deontic Logic, Answer Set Programming, Weak Contraints, Ethical AI} 
}
\begin{document}

\maketitle              

\begin{abstract}
The rise of powerful AI technology for a range of applications that are sensitive to legal, social, and ethical norms demands decision-making support in presence of norms and regulations.  
Normative reasoning is the realm of deontic logics, 
  that are challenged by well-known benchmark problems (deontic paradoxes), and lack efficient computational tools. 
  In this paper, we use Answer Set Programming (ASP) for addressing these shortcomings and 
  showcase how to encode and resolve several well-known deontic paradoxes utilizing weak constraints.
  %
  By abstracting and generalizing this encoding, we present a methodology for translating normative systems in ASP with weak constraints. This methodology is applied to "ethical" versions of Pac-man, 
%
where we obtain a comparable performance with related works, but ethically preferable results.
\end{abstract}

\section{Introduction}

Norms, which involve concepts such as obligation and permission, are
an integral part of human society.  They are enormously important in a
variety of fields -- from law and ethics to artificial intelligence
(AI).  
In particular, imposing norms -- be they ethical, legal or
social -- on AI systems is crucial, as these systems have become
ubiquitous in our daily 
routines. 

A main difference between norms and other constraints lies in the fact that norms typically allow for the possibility of violation.
Reasoning with and about norms
({\it normative reasoning}) requires deontic logic, the branch of logic that
deal with obligation and related concepts.  
Normative reasoning comes with a variety of idiosyncratic challenges, which are
often exemplified by benchmark examples (so called deontic paradoxes).  A crucial challenge is reasoning about sub-ideal situations, such as
{\em contrary-to-duty (CTD) obligations}, which are obligations only triggered by a violation.  Other challenges are associated, e.g., with defeasibility issues (norms having different
priorities, exceptions, etc.). 

The first deontic system introduced --
Standard Deontic Logic~\cite{Wri51} -- was failing most of the
benchmark examples, (un)deriving formulas which are counterintuitive
in a common-sense reading.  This has motivated the introduction of a
plethora of deontic logics, see, e.g.\ 
\cite{Hand}.
These logics have been investigated
mainly in connection with philosophy and legal reasoning, and with the
exception of Defeasible Deontic Logic DDL~\cite{DDL1,DDL2}, they lack
defeasibility and efficient provers.  Defeasibility and efficient
reasoning methods are instead offered by Answer Set Programming (ASP),
which is one of the most successful paradigms of knowledge representation and
reasoning for declarative problem solving~\cite{brew-etal-11-asp}. Indeed, in a long and systematic effort
of the knowledge representation community, efficient solvers for fast evaluation of ASP programs have been developed, see, e.g.,~\cite{DBLP:journals/aim/KaufmannLPS16}.
Defeasibility is also inherent in 
ASP, due to its default-negation and also weak constraints. 
This paper introduces a method for using weak constraints for encoding norms in ASP.
We first translate desired basic properties of deontic operators in a 
common core that will be used in all further encodings.
These properties are established by 
analyzing multiple well-known deontic paradoxes (e.g., 
\textit{Ross Paradox}, the \textit{Fence Scenario},\ldots).
By abstracting and generalising the encodings of the 
specific paradoxes, we provide a 
methodology for encoding 
normative systems in ASP with weak constraints. 
The methodology is put to work
on a case study from~\cite{N2019,CADE} that involves a
reinforcement learning agent playing a variant of the Pac-man
video game with additional ``ethical'' rules.
Our encoding is used as a ``shield" to filter out the non compliant actions of Pac-man and the outcome is compared with~\cite{CADE}, that uses DDL. 
For space reason, we must omit details, for which we refer
to~\cite{Master}.

\section{Preliminaries}


\paragraph{\bf Answer Set Programming.} 
%
%
 %
We consider extended logic programs with disjunction~\cite{gelf-lifs-91,leon-etal-2002-dlv}, which are finite sets of rules $r$ (henceforth referred to as programs)
\begin{align}
	H_1\lor \ldots \lor H_l \ruleatom A_1,\ldots A_n,\textit{ not }B_1,\ldots,\textit{ not }B_m. \quad l,m,n\geq 0, \label{rule}
\end{align}
where all $H_i$, $A_j$,  and $B_k$ are literals in a first-order language. Here {\it not} denotes weak (default) negation and $\neg$ (also written $-$) strong negation. Informally, $r$ can be read as: ``If all $A_i$ are true and for no $B_j$ there is evidence that it is true, at least one of $H_1,\ldots,H_l$ must be true''.

The answer sets of a ground (variable-free) program $\Pi$ are given in terms of 
consistent sets $S$ of ground literals as follows. Let $S\models r$ denote that 
$S\models B(r)$ implies $S\models H(r)$, where $S\models B(r)$ denotes that
$\{A_1,\ldots,A_n\}\subseteq S$ and $\{B_1,\ldots,B_m\}\cap S=\emptyset$, and $S\models H(r)$ denotes that $\{H_1\ldots,H_l\}\cap S\neq\emptyset$. $S$ is an answer set of $\Pi$ if $S$ satisfies all rules $r$ in $\Pi(S) = \{ r\in \Pi \mid S\models B(r)\}$,
and no $S'\subset S$ satisfies $\Pi(S)$. The answer sets of a general program $\Pi$ are those of its grounding $grd(\Pi)$ that consists of all ground instances of its rules.
%
%
For modeling defeasibility of obligations, we use weak constraints  of the form
\begin{align*}
	\wc  A_1,\ldots,A_n.\ [w:l]
\end{align*} 
(as in the system DLV\cite{leon-etal-2002-dlv}) where $A_1,\ldots,A_n$ are literals (that may be weakly negated) and $w,l\geq 0$ are the weight and the (integer) level of the weak constraint, respectively. Informally, weak constraints 
single out optimal answer sets that have minimal total weights of violated weak constraints with higher levels being more important; see \cite{leon-etal-2002-dlv} for formal definitions and more background.

\paragraph{\bf SDL.} Standard Deontic Logic~\cite{Wri51}, 
is the best known system of deontic logic. 

SDL formulas are constructed by the following grammar ($\mathcal{A}$ is the set of atomic propositions):\begin{align*}
	\varphi:=p\in\mathcal{A}\mid\neg\varphi\mid(\varphi)\mid\varphi\vee\varphi\mid\varphi\wedge\varphi\mid\varphi\rightarrow\varphi\mid O\varphi\mid P\varphi\mid F\varphi
\end{align*} 
SDL is a monadic deontic logic, as the operators $O$ (obligation), $P$ (permission) and $F$ (prohibition) apply to single formulas; they are read as ``it is obligatory that $\varphi$'',  ``it is permissible that $\varphi$'', and ``it is forbidden that $\varphi$'', resp., and inter-definable, e.g., $P\varphi := \neg O(\neg \varphi)$, and
$F\varphi := O\neg \varphi$.

The semantics of SDL 
-- also known as the modal logic KD --
is based on possible worlds,
where the accessibility relation is serial. A 
Hilbert system for SDL is obtained by adding the following axiom-schemata and rules to any axiomatization of classical propositional logic:
\begin{align*}
	& \text{If }\varphi\text{ is a theorem, }O\varphi\text{ is a theorem} && \text{(RND}) & &&\\
     & O(\varphi\rightarrow \psi)\rightarrow(O\varphi\rightarrow O\psi) && \text{(KD)}  & \phantom{1pt}O\varphi\rightarrow\neg O\neg \varphi && \text{(DD)}	
\end{align*}
In the following we will also use the derived axiom and rule:
\begin{align*}
  &  \text{If } \varphi \rightarrow \psi \text{ is a theorem, }O\varphi\rightarrow O\psi\text{ is a theorem} && \text{(RMD)} & \neg O\bot && \text{(OD)} 
\end{align*}


For more details about SDL and other deontic logics, see, e.g.,
\cite{Hand}


\section{Deontic Paradoxes} 
\label{Paradoxes}

As the foundation for our work, we examine several deontic paradoxes.
These consist of (un)derivable
formulas which are counter-intuitive 
when viewed from a common-sense perspective.
While referred to in the literature as paradoxes, many of 
them are not paradoxes \textit{per se}, 
but rather puzzles or dilemmas.
Deontic paradoxes play an important role in deontic logic and
normative reasoning; they serve as sanity checks for existing systems, and as driving force for defining  new systems. In particular, they 
exemplify 
that SDL fails to capture the nuances of normative reasoning expected in certain scenarios.
There are many such paradoxes. We categorise them according to the reason for their failure, see e.g.\
\cite{CarmoClass,Hand},  and analyze one example for each class:

\begin{enumerate}
	\item {Paradoxes centering around RMD:}
\textit{Ross's Paradox}, \textit{Good Samaritan Paradox}, \textit{\AA qvist's Paradox of Epistemic Obligation}
	\item {Puzzles centering around DD and OD:}
    \textit{Sartre's Dilemma}, \textit{Plato's Dilemma}
	\item {Puzzles centering around deontic conditionals:}
		\textit{Broome's Counterexample}, \textit{Chisholm's Contrary-to-Duty Paradox},
        \textit{Forrester's Paradox}, \textit{Considerate Assassin Paradox}, \textit{Asparagus Paradox},
        \textit{Fence Paradox}, \textit{Alternative Service Paradox}
\end{enumerate}
Deontic conditionals refer to obligations that arise situationally; written as $O(A/B)$ (to be read as ``$A$ is obligatory if $B$'') they have been introduced to cope with contrary-to-duty obligations, i.e., obligations which come into force when another obligation is violated.

\mypar{Paradoxes centering around RMD}
show, in general,  that SDL is too strong as it derives unwanted consequences.
An example is \emph{Ross's Paradox}, which consists of the following two sentences:
\begin{align*}
	&\text{It is obligatory that the letter is mailed.} 
 \tag{R1}\\
	&\text{It is obligatory that the letter is mailed or burned.} 
  \tag{R2}
\end{align*}
Let $O(m)$ and $O(m \vee b)$ formalize (R1) and (R2), respectively.
As $m\rightarrow(m\vee b)$ is a theorem in SDL,  
$O(m \vee b)$ follows from $O(m)$ by RMD and modus ponens. But it seems counterintuitive to derive an obligation that is satisfied by burning the letter, when failing to mail the letter.


\mypar{Puzzles centering around DD and OD} 
involve conflicting obligations that cannot be obeyed. 
An example is \textit{Plato's Dilemma}:
\begin{align*}
	&\text{It is obligatory that I meet my friend for dinner.}\tag{P1}\\
	&\text{It is obligatory that I rush my child to the hospital.}\tag{P2}
\end{align*}
Clearly, it is not possible to satisfy both obligations at the same time.
Using common sense reasoning, P2 should override P1, but in SDL the presence of two contradictory obligations  would make everything obligatory (deontic explosion). 

\mypar{Puzzles centering around deontic conditionals}
have as a prominent example the \textit{Fence Scenario}~\cite{CAP}, which combines two different weaknesses of SDL regarding CTD obligations and
exceptions:
\begin{align*}
	&\text{There must be no fence.}\tag{F1}\\
	&\text{If there is a fence then it must be a white fence.}\tag{F2}\\
	&\text{If the cottage is by the sea, there may be a fence.}\tag{F3}
\end{align*}
Here (F2) serves as a 
CTD obligation that is active
when 
(F1) is violated, while (F3) serves as an
exception to 
(F1). Note that under this
interpretation, if the cottage is by the sea the fence need not be white. 
The contrary-to-duty obligation (F2) cannot be formalised in SDL: having a white fence implies
having a fence, and by (RND) the obligation to have a fence; this
contradicts (F1). Moreover, as SDL lacks expressing defeasibility, (F3) cannot be properly formalized. 

\section{Encoding the Paradoxes}
\label{sec:encode-paradoxes}

We now proceed to encode the paradoxes from above All encodings share the same
common core, shown in Figure~\ref{fig: Core}, that encodes properties of SDL,
using the following predicates:
\begin{itemize}
	\item $O(X)$ resp.\ $F(X)$  denotes that $X$ is obligatory resp.\ forbidden;
	\item $act(X)$ denotes that $X$ is eligible for reasoning about whether $X$ is obligatory or not. 
    While an action by default, in some cases $X$ may not materialize but viewed as such. 
   An example from above would be owning a white fence, as we reason about whether it is obligatory.
	\item $Do(X)$ denotes that the agent has chosen to take the
          action $X$, and $-Do(X)$ denotes that the agent will
          definitely not take the action $X$.
	\item $Dia(X)$ is an auxiliary predicate to denote that an action $X$ is an option resp. possible (in the sense on modal logic).
    Thus, $-Dia(X)$ can either mean that the agent cannot take the action or that the agent has chosen not to take the action.\todo{Should we explain this further. TE: not sure. Perhaps we say "action $X$ is an option resp.\ possible" then this could be clearer. }
	\item $Happens(X)$ is an auxiliary predicate that denotes an event $X$ happening. It is sometimes used in encodings to denote events happening which are usually outside the agents control.
\end{itemize}
 \begin{figure}[t]\caption{The common core of our encodings}
 \centering
 		\boxed{\begin{aligned}
 				&O(X)\vee -O(X)\ruleatom act(X).&(1)&& \ruleatom F(X), Do(X). &&(7)\\
 				&F(X)\vee-F(X)\ruleatom act(X).&(2) &\qquad &Happens(X)\ruleatom Do(X).&&(8)\\
 				&\ruleatom O(X),-Dia(X).&(3) &
                &\ruleatom Do(X),-Dia(X).&&(9)\\
 				&-Dia(X)\ruleatom -Do(X), act(X).&(4)&
                 &\wc  O(X). [1:1]&&(10)\\
 				&\ruleatom O(X), F(X).&(5)& 
                &\wc  F(X). [1:1]&& (11)\\
 				&Do(X)\vee-Do(X)\ruleatom  act(X).&(6)&& && \\
 		\end{aligned}}
\label{fig: Core}
 \end{figure}
Intuitively, the common core guesses whether something is obligatory $(1)$, forbidden $(2)$ and whether the agent takes the action $(6)$. The remaining rules then encode connections between predicates and exclude answer sets that we deem inconsistent, e.g., something being obligatory and forbidden or something being obligatory and
that action not taken. The weak constraints $(10)$ and $(11)$ are used to eliminate answer sets that derive obligations/prohibitions with no need.


We note that the common core is not a faithful encoding of full SDL and its axioms: 
while theoretically possible, this would be undesired as SDL axioms lead to multiple paradoxes. Instead, the common core captures some of the SDL axioms, while satisfactory handling the deontic paradoxes. For instance
 
\newcommand{\propDD}{The SDL axiom $DD$ holds in the common core.}
\begin{proposition}\label{prop:DD} 
\propDD
\end{proposition}
\begin{proof}
 $DD$ is formalized in rule $(5)$, that forbids an action from being both forbidden and obligatory, as $O\varphi\rightarrow\neg O\neg\varphi$ which is equivalent to $\neg O\varphi\vee\neg F\varphi$, and to $\neg(O\varphi\wedge F\varphi)$.
\end{proof}

The idea behind the common core is to generate all maximal sets of non-auxiliary predicates that are consistent and filter out suboptimal answer sets using weak constraints. We require that in a consistent set an action is not obligatory and forbidden at the same time and any obligatory action is taken (resp. any forbidden action is not taken).\todo{Mention the auxiliary predicates as well? TE: don't think so.} 
For instance, disregarding the auxiliary predicates $Dia$, $Happens$ and $act$, the following are all maximal consistent sets of deontic predicates:
\begin{align*}
					&\{O(action), -F(action), Do(action)\},\{F(action), -O(action), -Do(action)\},\\
                    &\{-F(action), -O(action), -Do(action)\},\{-F(action), -O(action), Do(action)\}
				\end{align*}

The next proposition can be seen as a soundness and completeness result for the common core.

\newcommand{\propCorrect}{The rules $(1)$--$(9)$ from Figure~\ref{fig: Core} allow for all and only the maximal consistent sets of deontic predicates as answer sets.
}

\begin{proposition}\label{prop:correct} 
\propCorrect
\end{proposition}
This can be seen by inspecting the answers sets output by a solver. 
Intuitively, soundness of the rules $(1)$ to $(9)$ is achieved, as inconsistent answer sets are excluded. Completeness on the other hand is given, as all answer sets that are considered consistent for an action are generated.
%

Note that in Prop.~\ref{prop:correct}
the weak constraints $(10)$ and $(11)$ are excluded: as every answer set for them would represent an optimal way to handle given norms, they would eliminate every answer set containing obligations or prohibitions, and completeness would be lost.
By including them, preference will be given to unrestricted behaviour, in interplay with possible further rules. 


\subsection{Paradox encodings} 

We describe how to extend the common core to encode selected paradoxes from above. 

\mypar{Ross's Paradox.} 
In contrast with what happens in SDL, we do not want to derive (R2) in the ASP encoding.
To achieve this, we add  the following rule and facts to the core:
\begin{center}\boxed{\begin{aligned}
			&\wc -O(mail).\ [1:2]&(12) &\quad& 
			&act(mail).&(13)& \quad& 
			&act(burn).&(14)\\
	\end{aligned}}
 \end{center}
Note that a disjunction over obligations is represented by two different answer sets that each contain one possible way to satisfy the obligation over the disjunction.

The obligation $(R1)$ is created using the weak constraint $(12)$, while the facts $(13)$ and $(14)$ declare \textit{mail} and \textit{burn} as actions to reason about. For \textit{mail} (resp. \textit{burn}) the core encoding guesses it as obligatory or not; 
(12) may penalise the guess at the highest level if the program does not 
entail mailing the letter as obligatory. As 
constraint violation at the highest level is minimised,  each optimal answer set includes the obligation to mail the letter, should such an answer set exist.

The program has two answer sets; none of them derive the obligation to burn the letter, and they only differ for the choice of the agent to burn the letter or not. By adding a rule specifying that it is not possible to perform both acions: burn the letter and mail it, the answer set where the agent chooses to burn the letter would not be derived. 

\mypar{Plato's Dilemma.} 
Recall that the desired outcome of this dilemma would be that the agent takes her child to the hospital, thereby violating the obligation of meeting her friend for dinner. The encoding presents two interesting aspects: prioritisation of the obligations and the impossibility of taking both actions. This can be encoded as follows:
\begin{center}
	\boxed{\begin{aligned}
			&\wc -O(help), Happens(emergency). [1:3]&(20)& \qquad & act(help).&& (23)\\
			&\wc  -O(meet). [1:2]&(21)& &\ruleatom Do(help), Do(meet).&&(24)\\
			&act(meet).&(22)& &Happens(emergency).&&(25)
	\end{aligned}}\\\end{center}
  The weak constraint $(20)$ is at level $3$, the highest in this encoding, and penalises answer sets in which $Happens(emergency)$ is true but the obligation to help is not derived. In other words, it derives the obligation $(P2)$ to help the child in case of an emergency.
  The weak constraint $(21)$ encodes the obligation $(P1)$ to meet the friend for dinner, but at a lower level (viz.\ 2), which gives priority to $(P2)$.
%
		The constraint $(24)$ encodes the impossibility of taking both actions. 
  With the assertions that an emergency occurs and $meet$ and $help$ are the possible actions, 
  as desired, a single answer set exists containing the obligation to help the child.
  
\mypar{Fence Paradox} 
 One might think that CTD obligations could be handled like exceptions to obligations. While one could accordingly
 state ``There may be a fence if it is white'', it would not have the same meaning as in the paradox. Handling a CTD obligation as an exception results in a loss  the original obligation to a certain degree; 
 it could in this case be seen as {\em the least thing to do to set things right}. 
 While having a white fence improves the situation, the presence of the fence itself remains undesirable~\cite{CAP}.

	The important fact to consider is that should the cottage be by the sea, then as 
 (F1) is not active due to 
  (F3), the fence 
  need not be white.\label{reason} 
  To this end, we add the following to the common core:%
  \begin{center}%
			\boxed{\begin{aligned}
					&\wc  -F(have\_fence),\, not\, Location(sea).\ [1:2]&(30)&\\
					&\wc  Do(have\_fence),\, not\, Location(sea),& & &act(have\_fence).&&(32)\\[-2pt]
                    &\phantom{\wc }\ -O(have\_white\_fence).\ [1:2]&(31) & &act(have\_white\_fence).&&(33)
			\end{aligned}}\end{center}
	Here (30) caters for (F1) and (31) for (F2); notably,  
   (F3) also affects the CTD obligation (F2). This is needed as the fence has to be white only when the cottage is not by the sea. Otherwise the obligation for the fence to be white would also be derived if the cottage was by the sea.
		
  To check whether the obligation for the fence to be white is deduced when the cottage is by the sea (and we have a fence), 
  we further add:
  \begin{center}
			\boxed{\begin{aligned}
					&Location(sea).\qquad 
					&Do(have\_fence).
			\end{aligned}}\\\end{center}
		
Then two answer sets exist; both do not derive the obligation for the fence to be white. When testing other constellations, the answer sets obtained also represent the expected results.

 \section{Generalisation and Methodology}
		We can classify the obligations that appeared in the paradoxes
  into the following classes:
  \begin{itemize}
			\item \textbf{Regular obligations:} These obligations should be followed as long as possible (without violating a more important obligation).
			\item \textbf{Conditional obligations:} These are obligations that only need to be followed given certain preconditions. E.g., the obligation to wear a suit when at a formal event. 
			\item \textbf{Obligations over disjunctions:} obligations that are fulfilled by satisfying any disjunct that constitutes the obligation; 
             e.g., bring dessert or salad. 
			\item \textbf{Conjunctions of obligations that all need to be satisfied:} obligations consisting of multiple parts where satisfying all parts is necessary.
			\item \textbf{Obligations with exceptions:} 
            obligations to be followed unless an exception is given.
			\item \textbf{Contrary-to-duty obligations:} 
           obligations that arise as another obligation is violated. 
			\end{itemize}
		 Note that prohibitions are 
   viewed as regular {\em negative} obligations, i.e. the obligation not to do something. The different kinds of obligations are encoded as shown in Table~\ref{tab:encodings}.
	%
\begin{table}[t]
\renewcommand{\arraystretch}{1.1}
\centering		 	\begin{tabular}{ |c|l| } 
		 		\cline{1-2}
		 		Type of obligation & Encoding \\ 
		 		\cline{1-2}
		 		Regular &  $\wc -O(o).\ [w:l]$ \\ 
		 		Conditional & $\wc -O(o),condition.\ [w:l]$ \\ 
		 		Disjunction & $\wc -O(o_1),-O(o_2),\ldots,-O(o_n).\ [w:l]$\\
		 		Conjunction & $\wc  not\, Conj.\ [w:l]$  \\
		 		& $Conj:-O(o_1),\ldots, O(o_n).$\\
		 		Exceptions & $\wc -O(o),\phantom{o}not\, Exception.\ [1:2]$\\
		 		Contrary-to-duty & :$\sim-Do(o_1),-O(o_2).\ [1:2]$\\
		 		\cline{1-2}
		 	\end{tabular}
\caption{Encodings for different types of obligations}\label{tab:encodings}
\end{table}
%
Their encoding uses weak constraints
to model defeasibility. An obligation to take an action $a$ that should always hold is encoded in the following way:
\begin{align*}
			\wc  -O(a).\ [w:l]
		\end{align*}
Note that the weight $w$ and the level $l$ of the weak constraint depend on the importance of the obligation and conflicting obligations. In most cases, $w=1$ and merely $l$ is used to encode priorities among obligations. Conflicts between obligations are detected and the priority among obligations is established through weak constraints (more important obligations have higher level). 
		
		By generalising the encodings of the considered paradoxes (from Section~\ref{Paradoxes}), we propose the following encoding methodology 
       that consists of the following steps:
\begin{itemize} 
\item[] {\bf Step 1.} For each of the norms determine what kind of obligation it represents, among the six different kinds of obligations we have considered.
   
\item[] {\bf Step 2.} Determine which actions are simultaneously incompatible. 
			Knowing which actions are in conflict eases determining the importance of the obligations. Incompatibility needs to be determined through context. 
            E.g., while in general it is possible to watch a movie while browsing the internet, these actions may be incompatible 
            on an old 
            smartphone. 
            
\item[] {\bf Step 3.} Encode the different kinds of obligations and their importance. Here weights as priorities play an important role. 
			There are two cases we need to consider:
             
			\quad{\em Case 1}. 
             An obligation (of whatever kind) is more important than the other. In this case, setting the level of one constraint higher than 
             the other is sufficient. E.g., consider the case 
             of two obligations $o_1$ and $o_2$ where the latter is more important. If 
             $o_1$ and $o_2$ have no special properties, the encoding is: (with $j\geq 1$)
			\begin{align*}
\wc -O(o_1).\ [1:l] \qquad 
&\wc -O(o_2).\ [1:l+j]
			\end{align*}
			Note that $j$ may differ  given additional obligations.
			We account for this as follows.
   
			Suppose the incompatibilities between actions are given, as well as the importance of obligations by a preference $O'\succ O$ stating that the obligation $O'$ is strictly more important than $O$. We generate a directed graph $G=(V,E)$ whose vertices $V = \{O_1,\ldots,O_n\}$ are the obligations having the edges $E=\{(O_i,O_j) \in V^2\mid O_i\succ O_j\}$; note that $G$ must be acyclic.%
   \footnote{For Non-strict preference $\succeq$, we can use the supergraph of $G$, whose nodes cluster all equally preferable obligations.}
   %
   The sinks of $G$, i.e., vertices with no outgoing arcs, are assigned priority $p_1 = 2$. 
   After simultaneously removing all sinks, we iterate the process with increased priority, i.e., assign the new sinks 
   priorities $p_2=3$, $p_3=4$ etc.; this results in a priorization by levels.
			
			\quad{\em Case 2.} 
             Multiple obligations $o_1,\dots o_n$ are in conflict with an obligation $o$. While satisfying $o$ is better than satisfying a
             single $o_i$, satisfying multiple 
             $o_i$'s may be equally good or better than satisfying $o$. In this case, we can use the weights of the weak constraints for $o_1,\ldots, o_n$ to encode this: they must then correspond to their importance and $o$ must have a weight that is equal or smaller than the joint weights. 
			
			For example, if $o_1, o_2$ and $ o_3$ are mutually non-exclusive and equally important to $o$ if all are satisfied, the following weights could be chosen, for a number $k\geq 1$ (the level 
            is the same):\begin{align*}
&\wc -O(o_1).\ [k:l]\ &&\wc -O(o_2).\ [k:l]\
&\wc -O(o_3).\ [k:l]\ 
&&\wc -O(o).\ [3k:l]
			\end{align*}
\item[] {\bf Step 4.} Encode the exclusion of combinations of actions found incompatible. If two actions $a_1$ and $a_2$ are incompatible, this is encoded by adding the following constraint:
				\begin{align*}
						&\ruleatom Do(a_1), Do(a_2).
			\end{align*}
\item[] {\bf Step 5.} Encode additional information. This includes denoting constants as actions using the predicate $act$, and specifying dependencies between actions; 
%
e.g., if the action running entails the action moving, we add a rule
   \begin{align*}
Do(move)\ruleatom  Do(run).
			\end{align*}
\end{itemize}  
\begin{observation}\label{obs:RNDKD}
The common core simulates the SDL axioms $RND$ and $KD$.
\end{observation}
As there are no theorems in our framework, $RND$ is not directly implemented. However, if we read a theorem as something than cannot be violated, $RND$ becomes: ``If it is impossible to violate $\varphi$, then $\varphi$ is obligatory.'' In our semantics this obligation would not be derived; however the obligation to take the action $\varphi$ is given indirectly as the agent must do $\varphi$ (that cannot be violated). For $KD$, recall that if an obligation is in an answer set, the agent has to take the corresponding action. We encode  $O(\varphi\rightarrow\psi)$ using a weak constraint that sanctions answer sets including $Do(\varphi)$ but not $O(\psi)$. As every answer set that contains $O(\varphi)$ also contains $Do(\varphi)$, each answer set that contains $O(\varphi)$ but not $O(\psi)$ is also sanctioned. This way we simulate $KD$, as an answer set that contains $O(\varphi)$ must also contain $O(\psi)$, unless the latter conflicts with an obligation of higher or equal importance.
\section{A Case Study: Ethical Pac-man}\label{sec:use-cases}
We put the methodology described in the previous section to work 
on a reinforcement learning agent playing variants of
the game Pac-man. 
Pac-man features a closed environment with simple game mechanics and parameters which are easy to manipulate, and extend with norms that can simulate normative conflicts. 

\begin{wrapfigure}{R}{0.425\textwidth}
\vspace*{-0.5\baselineskip}
 \caption{Pac-man}
	\centering
	\includegraphics[width=0.4\textwidth]{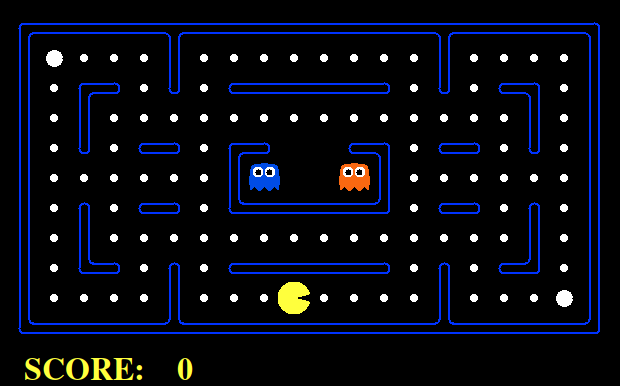}
	\label{fig:pacman}
\end{wrapfigure}
The starting position of the game is depicted in Fig.~\ref{fig:pacman}. Pac-man's objective is to eat all the pellets in the maze while avoiding two ghosts (orange and blue) that will kill him upon contact. Pac-man and the ghosts typically move one step at a time, with the ghosts' movements being non-deterministic. If Pac-man ate one of the larger pellets, the ghosts enter a scared state and become vulnerable, allowing Pac-man to eat them. In this state, the ghosts move at half speed. A scared ghost is instantly eaten by Pac-man when their distance is less than 1 on both axes. Points are awarded for consuming pellets and ghosts, with a discount applied based on the duration of the game (the longer the game lasts, the lower the score). Pac-man wins if he collects all the pellets, and a faster completion time results in a higher score. 
%

Following~\cite{N2019,CADE}, we consider variants of the Pac-man game with additional ``ethical norms". 
  {\em Vegan Pac-man}, in which Pac-man
is not allowed to eat any ghost has been introduced in~\cite{N2019} and implemented there using multi-objective Reinforcement Learning (RL) with policy orchestration. 
Vegan Pac-man and its vegetarian variant were analyzed in~\cite{CADE}, that will serve as the benchmark for our work.
  {\em Vegetarian Pac-man} can eat the orange ghost (as it
                would be cheese) but not the blue ghost. 
    %
    %
%
The approach used in~\cite{CADE} combines RL with formal tools for normative reasoning.
The authors implement a logic-based
\textit{normative supervisor} module, which informs the trained RL agent of the ethical requirements in force in a given situation. At each step, Pac-man  chooses an action complying with the norms, and a least evil action 
if there is no such action.
Their approach allows to deal with
complex normative systems, conflicting obligations, and situations where
no compliance is possible.  Norms and the current state of the agent's
environment are encoded in defeasible deontic logic~\cite{DDL1,DDL2},
 which is a deontic logic with 
defeasible rules that specify typical correlations,
such as  ``birds usually fly"; its theorem prover SPINdle is used to check norms compliance.  Exceptions are encoded by
so-called defeaters, e.g., if the bird is a penguin.
		
%
%
		
 We consider here an alternative realization of the normative
 supervisor, based on our norms encoding and using the DLV reasoner. We experimentally compare the obtained results and apply our methodology to more intricated ``ethical norms" for Pac-man.
		
\begin{description}
\item[Vegan Pac-man:]
To prohibit Pac-man from eating any ghost, we can state:
\begin{align*}
			O(\neg eat(g)) \qquad \text{ respectively
                        } \qquad F(eat(g)), \quad  \text{for } g \in \{ blue\_ghost, blue\_ghost \}.
		\end{align*}
\item[Vegetarian Pac-man:] To prohibit Pac-man from eating the blue ghost, we state:
		\begin{align*}
			O(\neg eat(blue\_ghost)) \qquad \text{
                          respectively } \qquad F(eat(blue\_ghost)).
		\end{align*}
\end{description}

We encode the norm bases by forbidding Pac-man to move in a direction if the ghosts are scared and 
moving there could lead to eat a ghost. Furthermore, we forbid Pac-man from stopping if  a ghost could move into Pac-man (and then be eaten).
%
Pac-man may still eat a ghost. This can happen if both a ghost and Pac-man move towards a larger pellet from perpendicular directions. In this case, Pac-man will first eat the pellet and then 
the ghost. 
Furthermore, Pac-man could be cornered between two scared ghosts, leaving
him no choice but eating one of them. 
	
The scenarios that can precede Pac-man eating a ghost are the same for both norm bases. As Pac-man and the ghosts can move at most one step at a time, we can derive that the Manhattan distance between Pac-man and a scared ghost must in this case be
at least 1 and at most 2 (coordinates are integers). This leaves three possibilities for their relative locations. We encoded the norms by
accounting for the locations of the ghosts relative to Pac-man and 
forbidding Pac-man to make moves that could lead to
eating a ghost.

\mypar{Experimental Results.}
The vegan norm base was implemented in~\cite{N2019}  who trained two different models for Pac-man;
in one model he was trained to maximize the game score, and in the other to comply with the norms using respective data. 
 An external function enabled the agent to decide which model to use for choosing the next move. When the importance of the norm-compliance model was low, the agent in general did not comply with the norm, resulting in around $2$ ghosts eaten per game. 
Making the importance sufficiently high, the number of ghosts eaten did decrease to one (similar to our results, see Table~\ref{tab:results}). In this approach, however, it is not clear how to enforce more complex norm sets, including, e.g., contrary-to-duty obligations.

In order to compare Neufeld et al.'s approach~\cite{CADE} with our ASP encodings, we utilized their Java-based framework. This framework is built in Java and provides the normative reasoners with a ready-to-use interface to the Pac-man game.%
\footnote{\label{fn:code}For complete DLV code and experimental data, see {\scriptsize\url{github.com/Chrisi-boop/DLV-Normative-Reasoning-}}.}
 The  RL agent was trained on $250$
                games and the normative supervisor was evaluated
                on $1000$ games with the same initial state.

\begin{table}[t]
\centering
\begin{tabular}[h]{@{}l*{4}{|@{\!}c@{}c@{}c}@{}}
\renewcommand{\arraystretch}{1.1}
	{norm base} &\multicolumn{3}{c|@{\!}}{%
     \hspace*{-2pt}\% games won} &\multicolumn{3}{c|@{\!}}{\hspace*{-2pt}game score avg[max]} &\multicolumn{3}{c|@{\!}}{\hspace*{-2pt}avg ghosts eaten (blue/orange)} &\multicolumn{3}{c@{}}{\hspace*{-2pt}avg time (s)} \\[2pt]
 \cline{1-13}
Vegan & ~$90.7$~ & |& $91.2$ & $~1209.86[1708]~$ &|& $~1217[1538]$ & ~~$0.023/0.02$~ &|& $0.013/0.018$ & ~10.1~ &|& 6.7\\
Vegetarian & $94.0$ &|& $90.6$ & $~1413.80[1742]~$ &|& $~1366[1751]$ & $\phantom{0}0.01/0.79$ &|& $0.001/0.788$ &  9.8 &|& 6.5\\
Weak Vegan &  &|& $89.9$ &  &|& $~1204[1731]$ &  &|& $0.002/0.043$ &  &|& 6.7
\end{tabular}

\caption{Results for Neufeld et al.'s normative supervisor~\cite{CADE} ~|~  our encoding}\label{tab:results}
\end{table}
The results reported in~\cite{CADE} for their normative supervisor (using
SPINdle)
and for our ASP-based normative supervisor
(using JDLV, the java framework for DLV)
are shown in Table~\ref{tab:results}.
The latter outperformed the original supervisor for both norm bases
w.r.t.\ the ghosts eaten and 
speed: its average running time 
was less than 7~seconds, while
for Neufeld et al.'s normative supervisor it was some seconds 
(roughly 50\%) longer,
 using a Lenovo Y50-70 with 8GB RAM and Ubuntu 22.04 LTS.
For the vegan norm base, our encoding resulted in a higher winning rate and an improved average score. (A game loss costs $500$ points, resulting possibly in a negative score.)
The higher maximum score 
in Neufeld et al.'s results 
 is likely attributed to a game where both ghosts were eaten.
 In the vegetarian norm base, the lower winning rate and average score is probably due to Pac-man's preference to lose in our framework a game rather than eating a ghost.


\mypar{Weak Vegan Pac-man.} Consider the following more intricated 
variant of the vegan norm base:
\begin{enumerate}
                	\item[O1] It is obligatory not to eat the blue ghost.
                	\item[O2] It is obligatory not to eat the orange ghost, unless a ghost has already been eaten. 
                	\item[O3] It is obligatory to stop (for one move) after having eaten a ghost. 
                 \end{enumerate}
Thus, the behaviour may change at some point. 
To encode it, we follow our methodology:

\noindent\textbf{Step 1:} We categorise the obligations.
\begin{myitemize}
        \item O1 is a regular obligation. It 
        has no exceptions and 
        should be followed at all times.
        \item O2 is an obligation with an exception. Once the latter occurs, it holds until the game is over.
        \item O3 is a derived obligation. Note that it is not necessarily a CTD obligation, as eating a ghost may not violate an obligation if the ghost is blue.
    \end{myitemize}
\noindent\textbf{Step 2:} We look at pairs of obligations that cannot be fulfilled simultaneously.
\begin{myitemize}
        \item O1 and O2 could possibly be in conflict, as Pac-man may be stuck between two ghosts moving towards him. This may force him to eat one of them. In this case we prioritize O1, as we want preventing the blue ghost from being eaten to have the highest priority.
        \item O1 and O3 could be in conflict, as stopping after eating an orange ghost 
        may lead to eating a blue ghost. 
        \item O2 and O3 cannot be in conflict, as after eating a ghost the exception to O2 is empowered. 
        We give priority to O1 over O3 since the preservation of the blue ghost from being consumed  has highest priority. 
    \end{myitemize}
    Summarizing the statements above, we obtain the following preferences:
%
                	O1  $\succ$  O2, 
                	O1  $\succ$  O3.
                
\noindent\textbf{Step 3:} 
 We derive the following weights and levels for the weak constraints
                corresponding to the obligations: O1 is assigned [1 : 3],  O2 is assigned [1 : 2], and O3 is assigned [1 : 2].

 We next look at the predicates used in the encoding:%
 \begin{itemize}
            \item \textit{Pacman(X,Y)} 
            denotes the location of Pac-man. $X$ and $Y$ refer to his coordinates on the map.
            \item \textit{BlueGhost(X,Y,B)} and \textit{OrangeGhost(X,Y,B)} 
            denotes the location of the blue resp.\ orange ghost, where $X,Y$ are
            its coordinates and $B$ 
            is $1$ if the ghost is scared and $0$ otherwise. 
        
            \item \textit{Exception} means the agent has already eaten a ghost and 
            may eat orange ghosts. 
            Whenever {\it Exception} is in an answer set, the normative reasoner will inject it into any later answer set.
        \end{itemize}
        The actions to reason about are \textit{stop}, \textit{north}, \textit{east}, \textit{south}, \textit{west}, \textit{eat(blue\_ghost)}, \textit{eat(orange\_ghost)}:\begin{itemize}
            \item \textit{Do(stop)} means that the agent remains stationary for one action.
            \item \textit{Do(d)}, $d \in \{ north, east, south, west\}$, means that the agent will move in direction $d$.
            \item \textit{Do(eat(g)}, $g \in \{ blue\_ghost, orange\_ghost\}$ means that the agent eats 
            that ghost.
          \end{itemize}
\begin{figure}[tb]
\caption{Encoding for the Pac-man norm base}\label{fig:Pacmannorm}
                	\begin{subfigure}{0.3825\textwidth}
                		\centering
                		$\begin{aligned}
                				&\wc  -F(eat(blue\_ghost)).\ [1:3]\\
                				&\wc  not\phantom{o}Exception, \\[-2pt]
                                & \phantom{\wc \;\; } -F(eat(orange\_ghost)).\ [1:2]\\
                				&\wc   -O(stop), \\[-2pt]
                                & \phantom{\wc \;\; } Do(eat(orange\_ghost)).\ [1:2]\\
                                &\wc  -O(stop), \\[-2pt]
                                & \phantom{\wc \;\; } Do(eat(blue\_ghost)).\ [1:2]\\[4pt]
                                ~\\
                		\end{aligned}$
                		\caption{obligations \\ ~}\label{fig:Pacman-a}
                	\end{subfigure}
               	\begin{subfigure}{0.33\textwidth}
                		\centering
                		$\begin{aligned}
                				&\ruleatom Do(stop),\phantom{o}Do(north).\\[-2pt]
                                &\qquad\cdots\\[-2pt]
                                &\ruleatom F(eat(blue\_ghost)),\\
                                &\phantom{\ruleatom }\, {-}F(east), Pacman(A,B),\\ &\phantom{\ruleatom }\; BlueGhost(C,D,1),\\
                               &\phantom{\ruleatom }\; E=C-A,\, E\leq 2,\\ 
                                &\phantom{\ruleatom }\; G=B-D,\, G\leq1.\\[-2pt]
                                &\qquad \cdots\\
                                ~\\[-0.7\baselineskip]
                		\end{aligned}$                        
                		\caption{conflicting actions \\ ~ }\label{fig:Pacman-b}
                	\end{subfigure}
              	\begin{subfigure}{0.275\textwidth}
                		\centering
                		$\begin{aligned}
                				&act(d). \text{ for }d \in D\\
                				&act(eat(blue\_ghost)). \\
                			      &act(eat(orange\_ghost)).\\
                                &\ruleatom \; \textstyle\bigwedge_{d\in D} F(d).\\
                                ~ \\
                                ~ \\
                                ~ \\  
                                ~\\[-0.1\baselineskip]
                		\end{aligned}$
                		\caption{action 
                      specs; $D = \{ stop,$ $north,$ $east,$ $south,$ $west \}$}\label{fig:Pacman-c}
                	\end{subfigure}
\end{figure}                
To encode the norms we add the rules and facts in 
Fig.~\ref{fig:Pacmannorm} to the common core.
                Following our methodology, we encode the obligations as in Fig.~\ref{fig:Pacman-a}. 
                 Note that the obligation to eat the blue and orange ghost are encoded as separate actions as two different obligations (O3).

\noindent\textbf{Step 4:} 
The conflicting actions from Step~2 are encoded in Fig.~\ref{fig:Pacman-b}.
For Pac-man, we encode that he cannot take two actions at once, e.g., stop and move north.
For ghosts, as a scared ghost may move both towards and away from Pac-man, we cannot encode ensuing conflicts directly. 
However, we can assert that permitting an action that could potentially
involve consuming a ghost is inconsistent with the prohibition of
consuming that ghost.
Fig.~\ref{fig:Pacmannorm}.b shows one such rule; for space reasons we omit the other rules, as well as further action constraints for Pac-man.$^{\ref{fn:code}}$
                	
\noindent\textbf{Step 5:} 
We state the actions to reason about (Fig.~\ref{fig:Pacman-c}),
and encode that it is not possible to prevent Pac-man from taking any action (i.e., Pac-man must either stop or move).

 In fact, our actual encoding omits the actions $eat(blue\_ghost)$ and $eat(orange\_ghost)$; they are substituted by prohibitions on moving towards scared ghosts. We used them here for readability.            

\mypar{Experiments.}
The implementation of this norm base yielded results that align with our expectations, as depicted in Table~\ref{tab:results}. Pac-man consumed more ghosts compared to the vegan norm base but fewer than the vegetarian norm base, and mostly orange ghosts. The obligation to halt when Pac-man eats a ghost resulted in a lower average score, as time passage leads to point deductions. This observation also explains the marginal increase in the maximum score and the slight decrease in the winning rate. Despite the added complexity of the norm base, there was no significant difference in the running time.


\section{Related Work and Conclusion}

Starting with an analysis of well-known deontic paradoxes, we have introduced a methodology to encode normative systems in ASP, using DLV as the system of choice. Our approach determines optimal ways to handle scenarios, using agreed upon prioritizations of obligations.
  
Multiple approaches to implement normative systems do exist. Some of those related to the multi-agent systems community can be found, e.g., in~\cite{alechina}. We will discuss below the approaches most similar to ours.
		
One of the earliest works on encoding normative systems is Sergot et al.~\cite{Nation}, who did encode the British Nationality Act in logic programming. 
However, their work focused on determining the applicability of the British Nationality Act to specific individuals, without delving into the reasoning about obligations or seeking optimal courses of action within the framework of given norms.

\cite{Impact} introduced in IMPACT syntax and semantics for reasoning about obligations and prohibitions among agents in a rule-based language under different semantics, among them stable model semantics. Although they refer to deontic logic, the proposed way of dealing with conflicting obligations is to satisfy a maximal subset of obligations, without considering possible preferences among them. 
		
       An abductive logic programming framework has been employed in~\cite{Goal} to encode 
       obligations in various deontic paradoxes. Rather than trying to derive all
                optimal ways of fulfilling given obligations, the authors
                focused on finding a best model of a definite Horn logic program that satisfies given goals.
                When it comes to establishing model preferences, the auxiliary symbols used in the encodings played a significant role.
                However, their usage requires care, and no systematic way of defining model preference was considered, for which our work could provide inspiration. 
                
		Using a combination of input-output logic and game
                theoretic methods, ~\cite{regcon} 
                encoded
                the behaviour of (multi-)agents 
                subjected to a normative system. In their work,
                agents possess the ability to reason in a manner more akin to humans, for example, determining if the consequences of violating an obligation are worth the penalty. 
                Their work lacks however computational support.

		Temporal Logic has been used to simulate normative reasoning. 
  Temporal Logic approaches offer an advantage in enforcing norms that indirectly restrict certain actions~\cite{LTL1,LTL2}. Moreover there are advanced tools that effectively combine reinforcement learning, with, e.g., LTL (see \cite{ABEKNT2018}), and LDLf (see \cite{GiacomoV15}). 
However, as is well known temporal logic cannot handle all the intricacies of normative reasoning, see, e.g., the discussion in~\cite{alechina,NBC2022}.

		Different ASP approaches have been proposed in~\cite{DLP,AlferesGL13,tempdeon}. In these works, logic programs were extended with the SDL modality, and~\cite{tempdeon} introduces also temporal operators. Alferes et al.~\cite{AlferesGL13} further provides a way to check equivalence of two deontic logic programs. These approaches however
   require an understanding of the embedded logic and some implementation efforts, whereas our method can be used out of the box and with an ASP solver off the shelf. 
		

                                

In our approach, the selection of paradoxes has proven to be of utmost importance,  as failure to include certain ones resulted in overlooking crucial aspects of normative systems; 
e.g., the analysis of the \textit{Fence Paradox} 
has enabled us to differentiate between contrary-to-duty obligations and exceptions, which is a well-known problem in the field of normative reasoning~\cite{CAP}.
%
An important feature of our approach is the availability of optimized
tools (ASP solvers) and the simplicity of the  encoding. There is a
clear cut common core that is supplemented with defined ways for encoding different kinds of obligations.
The approach, described for DLV in this work, can be easily
transferred to other ASP solvers, e.g.\ clasp ({\small\url{potassco.org/clasp}}).
%
A main limitation of our approach is that
encoding intricate normative systems can result in large programs.
		Furthermore, obligations that require the agent to maintain certain conditions might be cumbersome to encode for complex scenarios where taking
                particular actions might indirectly lead to a violation.  
                Consider for instance
                ``See to it that the child stays dry.''
                Here it is not enough to simply not get the child wet, but one must also take measures to protect the child from getting wet through other means. This may entail preventing the child from going outside if it is about to rain.
                For such normative systems, 
                approaches extending ASP with temporal logics
                \cite{tempdeon}, 
                may be preferable. 

\mypar{Outlook.} 
A peculiar feature of our approach is that all encoded obligations are comparable based on
their associated weights and priorities. 
While in our encoding we did not run into problems,
 it is possible that certain normative systems may require optimal answer sets encoding solutions that are inherently incomparable.
 Future work may also look into other ASP  solvers with more sophisticated means for filtering out answer sets to model normative reasoning. The clasp extension asprin 
({\small\url{potassco.org/asprin}})
or the DLV$2$ system using  WASP \cite{WASP}, could serve as examples.
		
The Pac-man case study highlighted the fact that even simple obligations can require a substantial number of weak constraints to accurately encode the means of fulfilling those obligations. To address this weakness, we plan to explore the integration of our framework with auxiliary software capable of reasoning about the actions neeeded to fulfill specific obligations. In the context of Pac-man, 
such software might have to interpret, e.g., how the obligation to not eat a ghost could be fulfilled.

\bibliography{Quellen}
\bibliographystyle{eptcs}


\end{document}